\documentclass{CRPITStyle}

\usepackage{algorithm}
\usepackage{algpseudocode}
\usepackage[pdftex]{graphicx}
\usepackage{color}
\usepackage{transparent}

\usepackage[authoryear]{natbib}
\renewcommand{\cite}{\citep}
\newtheorem{theorem}{Theorem}[section]
\newtheorem{lemma}[theorem]{Lemma}

\newenvironment{proof}[1][Proof]{\begin{trivlist}
\item[\hskip \labelsep {\bfseries #1}]}{\end{trivlist}}
\newenvironment{definition}[1][Definition]{\begin{trivlist}
\item[\hskip \labelsep {\bfseries #1}]}{\end{trivlist}}
\newenvironment{example}[1][Example]{\begin{trivlist}
\item[\hskip \labelsep {\bfseries #1}]}{\end{trivlist}}

\newcommand{\qed}{\nobreak \ifvmode \relax \else
      \ifdim\lastskip<1.5em \hskip-\lastskip
      \hskip1.5em plus0em minus0.5em \fi \nobreak
      \vrule height0.75em width0.5em depth0.25em\fi}

\newcommand\conferencenameandplace{Thirty-Seventh Australasian Computer Science Conference (ACSC2014), Auckland, New Zealand, January 2014} \newcommand\volumenumber{147} \newcommand\conferenceyear{2014} \newcommand\editorname {Bruce H. Thomas and David Parry}

\begin{document}

\title{Combining the Shortest Paths and the Bottleneck Paths Problems}
\author{Tong-Wook Shinn$^1$
\and 
Tadao Takaoka$^2$}
\affiliation{Department of Computer Science and Software Engineering \\
University of Canterbury \\
Christchurch, New Zealand \\
$^1$ Email:~{\tt tad@cosc.canterbury.ac.nz} \\
$^2$ Email:~{\tt tong-wook.shinn@pg.canterbury.ac.nz}}

\maketitle

\toappear{This research was supported by the EU/NZ Joint Project, Optimization and its Applications in Learning and Industry (OptALI).}
\toappearstandard

%Include either the full copyright statement:
%\toappear{Copyright \copyright 2006, Australian Computer Society, Inc.  This paper appeared at the Seventeenth Australasian Database Conference (ADC2006), Hobart, Australia.  Conferences in Research and Practice in Information Technology (CRPIT), Vol. 49. Gillian Dobbie and James Bailey, Eds. Reproduction for academic, not-for profit purposes permitted provided this text is included. }

%or more easily (and recommended) use the alternative:

%\newcommand\conferencenameandplace{Twenty-Ninth Australasian Computer Science Conference (ACSC2006), Hobart, Australia}
%\newcommand\volumenumber{48}
%\newcommand\conferenceyear{2014}
%\newcommand\editorname{Vladimir Estivill-Castro and Gillian Dobbie}
%\toappearstandard 

%  For Government work where the alternative copyright statement will be required, please replace the \toappearstandard clause by (uncommented)
%\newcommand\copyrightholder{Commonwealth of Australia} - or whoever
%\toappeargovwork

\begin{abstract}
We combine the well known Shortest Paths (SP) problem and the Bottleneck Paths (BP) problem to introduce a new problem called the Shortest Paths for All Flows (SP-AF) problem that has relevance in real life applications. We first solve the Single Source Shortest Paths for All Flows (SSSP-AF) problem on directed graphs with unit edge costs in $O(mn)$ worst case time bound. We then present two algorithms to solve SSSP-AF on directed graphs with integer edge costs bounded by $c$ in $O(m^2 + nc)$ and $O(m^2 + mn\log{(\frac{c}{m})})$ time bounds. Finally we extend our algorithms for the SSSP-AF problem to solve the All Pairs Shortest Paths for All Flows (APSP-AF) problem in $O(m^{2}n + nc)$ and $O(m^{2}n + mn^{2}\log{(\frac{c}{mn})})$ time bounds. All algorithms presented in this paper are practical for implementation.
\end{abstract}
\vspace{.1in}

\noindent {\em Keywords:} Shortest Paths, SP, Bottleneck Paths, BP, Single Source Shortest Paths, SSSP, All Pairs Shortest Paths, APSP

%%%%%%%%
%%%%%%%% NEW SECTION 
%%%%%%%%

\section{Introduction}
The problem of finding the shortest paths between pairs of vertices on a graph is one of the most extensively studied problems in algorithm research. This problem is formally known as the Shortest Paths (SP) problem and is often categorized into the Single Source Shortest Paths (SSSP) problem and the All Pairs Shortest Paths (APSP) problem. As the names suggest, the SSSP problem is to compute the shortest paths from one single source vertex to all other vertices on the graph, and the APSP problem is to compute the shortest paths between all possible pairs on the graph. The most well known algorithm for solving the SSSP problem is the algorithm by \citet{Dijkstra} that runs in $O(n^2)$ time, where $n$ is the number of vertices in a graph. This algorithm can be enhanced with a priority queue, and if the Fibonacci heap \cite{FT}  is used to implement the priority queue then the time complexity becomes $O(m + n\log{n})$ where $m$ is the number of edges in the graph. If edge costs are integers bounded by $c$, then the SSSP problem can be solved in $O(m + n\log{\log{c}})$ time \cite{Thorup} using a complex priority queue for integers. For solving the APSP problem, the $O(n^{3})$ algorithm by \citet{Floyd} is the most well known.

If edges have capacities, the bottleneck of a path is the minimum capacity out of all edge capacities on the path. In other words, the bottleneck of a path from vertex $u$ to vertex $v$ is the maximum amount of flow that can be pushed from $u$ to $v$ down the path. Finding the paths that give the maximum bottlenecks between pairs of vertices is also a well studied problem and is formally known as the Bottleneck Paths (BP) problem. The Single Source Bottleneck Paths (SSBP) problem can be solved with a simple modification to the algorithm by \citet{Dijkstra}. For an undirected graph the All Pairs Bottleneck Paths (APBP) problem can be solved in $O(n^{2})$, which is optimal \cite{Hu}.

The SP and BP problems are concerned with finding the minimum or the maximum possible values. However, the shortest path may not give the biggest bottleneck, and the path that gives the maximum bottleneck may not be the shortest path. If the flow demand from a vertex to another vertex is known, then it is clearly beneficial to find the shortest path that can fully accommodate that flow demand. Thus we combine the SP and BP problems to compute the shortest paths for all possible flow amounts. We call this problem the Shortest Paths for All Flows (SP-AF) problem. As is common in graph paths problems, we divide the SP-AF problem into the Single Source Shortest Paths for All Flows (SSSP-AF) problem, and the All Pairs Shortest Paths for All Flows (APSP-AF) problem.

There are many obvious real life applications for this new problem, such as routing in computer networks, transportation, logistics, and even planning for emergency evacuations. The city of Christchurch has recently been hit by a series of strong earthquakes, most notably in September 2010 and in February 2011. If we consider hundreds of people evacuating from a building in such emergencies, if the amount of people (flow) can be predetermined, we can find the shortest route for all people in various locations around the building to their respective evacuation points such that the flow of people can be fully accommodated. If the flow amounts (of people) are not considered in calculating the evacuation routes, congestion may occur at various points in the building that could lead to serious accidents.

In this paper we present algorithms for solving the SSSP-AF and APSP-AF problems on directed graphs with non-negative integer edge costs and real edge capacities that are faster than the straightforward methods. We first give an algorithm to solve SSSP-AF on graphs with unit edge costs in $O(mn)$ time. We then present two algorithms for solving SSSP-AF on graphs with non-negative integer edge costs of at most $c$ in $O(m^2 + nc)$ and $O(m^2 + mn\log{(\frac{c}{n})})$ time bounds. Finally we show that the main concepts behind the SSSP-AF algorithms can be extended to solve APSP-AF in $O(m^{2}n + nc)$ and $O(m^{2}n + mn^{2}\log{(\frac{c}{mn})})$ time bounds.

%%%%%%%%
%%%%%%%% NEW SECTION 
%%%%%%%%

\section{Preliminaries}
\label{sec:prem}
Let $G=\{V,E\}$ be a directed graph with non-negative integer edge costs bounded by $c$ and non-negative real capacities. Let $n = |V|$ and $m = |E|$. Vertices are given by integers such that $\{1,2,3,...,n\} \in V$. Let $(i,j)$ be the edge from vertex $i$ to vertex $j$. Let $cost(i,j)$ denote the edge cost (or distance) and let $cap(i,j)$ denote the edge capacity. Let $OUT(v)$ be the set of vertices that are directly reachable from $v$, and $IN(v)$ be the set of vertices that can directly reach $v$. There can be up to $m$ distinct capacities if all edge capacities are unique. We refer to the distinct capacities as maximal flows. Then the SP-AF problem is to solve the SP problem for all maximal flows.

For all our algorithms only comparison operations are performed on the maximal flow values. Therefore the real values of maximal flows can be mapped to integer values without any loss of generality by first sorting the maximal flows in increasing order then assigning incremental integer values starting from 1. This allows us to use maximal flow values as indexes of arrays in our algorithms.

We use the computational model that allows comparison-addition operations and random access with $O(\log{n})$ bits to be performed in $O(1)$ time.

%%%%%%%%
%%%%%%%% NEW SECTION 
%%%%%%%%

\section{Single Source Shortest Paths for All Flows}
\label{sec:ssspaf}
We first consider solving the SP-AF problem from a source vertex $s$ to all other vertices in $G$. Initially we compute just the distances rather than actual paths then  later show that the paths information can be computed in the same time bound. That is, we first solve the Single Source Shortest Distances for All Flows (SSSD-AF) problem, then show that the algorithm can be extended to solve the SSSP-AF problem with no increase in the worst case time complexity.

The SSSD-AF problem can be defined as the problem of computing the set of all $(d,f)$ pairs for each destination vertex, where $d$ is the shortest distance from $s$ and $f$ is the maximal flow value. Let $S[v]$ be the set of $(d,f)$ pairs for the destination vertex $v$. Suppose $(d,f)$ and $(d',f')$ both exist in $S[v]$ such that $d < d'$. Then we keep $(d',f')$ \emph{iff} $f < f'$, that is, a longer path is only \emph{useful} if it can accommodate a greater flow. If $d = d'$, we keep the pair that gives us the greater flow.

\begin{figure}
\def\svgwidth{230pt}
\begin{center}
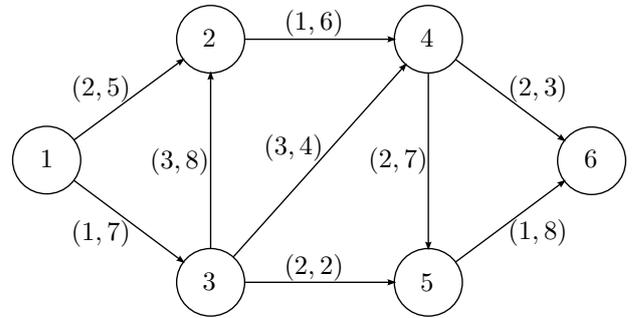
\end{center}
\caption{An example of a directed graph with $n=6$, $m=9$ and $c=3$. The first number in the parenthesis is the edge cost and the second number is the edge capacity.}
\label{fig:graph}
\end{figure}

\begin{example}
Solving SSSP-AF on the example graph in Figure \ref{fig:graph} with $s=1$ would result in $S[6] = \{(4,2),(5,3),(6,5),(8,6)\}$.
\end{example}

The straightforward method for solving the SSSP-AF problem is to iterate through each maximal flow $f_{i}$ and solving the SSSP problem for the sub-graph that only have edges with capacities $f_{i}$ or greater. On graphs with unit edge costs SSSP can be solved in $O(m)$ time with a simple breadth-first-search (BFS), resulting in $O(m^2)$ time bound for solving the SSSP-AF problem. On graphs with integer edge costs bounded by $c$, we can use the algorithm by \citet{Thorup} to solve SSSP-AF in $O(m^2 + mn\log{\log{c}})$ time.

Note that SSSP-AF cannot be solved with a simple decremental algorithm, where edges are removed one by one in decreasing order of capacity then the connectivity of all affected vertices are checked. This method fails because edges with larger capacities may later be required to provide shorter paths for smaller maximal flows.

\subsection{Unit edge costs}
\label{sec:ssspaf:unit}
Algorithm \ref{alg:unit} solves SSSP-AF in $O(mn)$ time by utilizing the fact that even though there are $O(m)$ maximal flows, there can only be $O(n)$ paths with unique path costs for each destination vertex. As noted earlier, if multiple paths from $s$ to $v$ exists with equal path costs, we only need to keep the path that can accommodate the biggest maximal flow out of those paths. Thus for graphs with unit edge costs, even with $m$ maximal flows, the size of $S[v]$ is $O(n)$ for each $v$.

Let $B[v]$ be the bottleneck of a path from $s$ to vertex $v$. Let $D[v]$ be a possible distance from $s$ to $v$. Let $Q[i]$ be a set of vertices that may be added to Spanning Tree (SPT) at distance $i$, such that $1 \leq i \leq n - 1$, i.e. one set of vertices exists for each possible distance from $s$.

The algorithm starts off with just $s$ in the SPT and makes incremental changes to the SPT as we iterate through the maximal flow values in increasing order. The SPT is a persistent data structure and we do not build it up from scratch in each iteration. In summary, as we iterate through each maximal flow we cut nodes from the SPT that cannot accommodate the maximal flow and add the nodes back to the SPT at the shortest possible distance from $s$ (the root).

\begin{algorithm}
\caption{Solve SSSD-AF on graphs with unit edge costs in $O(mn)$ time}
\label{alg:unit}
\begin{algorithmic}[1]
\algnotext{EndFor}
\algnotext{EndIf}
\algnotext{EndWhile}
\ForAll{$v \in V$}
	\State{$B[v] \leftarrow 0, D[v] \leftarrow 0$}
\EndFor
\State{$B[s] \leftarrow \infty$, $SPT \leftarrow s$}
\ForAll{maximal flows $f$ in increasing order}
	\ForAll{$v \in V$ such that $B[v] < f$}
		\If {$v$ is in SPT}
			\State{Cut $v$ from $SPT$}
		\EndIf
		\State{$D[v] \leftarrow D[v] + 1$}\label{line:inc1}
		\State{Add $v$ to $Q[D[v]]$}
	\EndFor
	\For{$i \leftarrow 1$ to $n-1$}
		\While{$Q[i]$ is not empty}
			\State{Remove $v$ from $Q[i]$}
			\ForAll{$IN(v)$ as $u$}
				\If{$D[u] = D[v] - 1$}\label{line:D}
					\State{$b \leftarrow \Call{min}{cap(u,v),B[u]}$}
					\If{$b > B[v]$}\label{line:B}
						\State{$B[v] \leftarrow b$}
						\State{Add $v$ to $SPT$, $u$ as parent}
					\EndIf
				\EndIf
			\EndFor
			\If{$v$ is in $SPT$}
				\State{Append (D[v],B[v]) to $S[v]$}\label{line:append}
			\Else
				\State{$D[v] \leftarrow D[v] + 1$}\label{line:inc2}
				\State{Add $v$ to $Q[D[v]]$}
			\EndIf
		\EndWhile
	\EndFor
\EndFor
\end{algorithmic}
\end{algorithm}

\begin{lemma}
Algorithm \ref{alg:unit} correctly solves SSSD-AF on directed graphs with unit edge costs.
\end{lemma}
\begin{proof}
By iterating from $Q[1]$ to $Q[n-1]$ for each maximal flow, we ensure that all vertices are added to the SPT at the minimum possible distance from $s$ for the maximal flow value in the current iteration. When the minimum possible distance is found for a vertex $v$ to be added to the SPT (line \ref{line:D}), all potential parent nodes, $IN(v)$, are inspected to ensure that $v$ is added to the $SPT$ such that the bottleneck from $s$ to $v$ is maximized for the given distance (line \ref{line:B}). Since $D[v]$ is monotonically increasing (lines \ref{line:inc1} and \ref{line:inc2}), $v$ cannot be added to the SPT multiple times at the same distance. Thus any time a vertex $v$ is added to the SPT, $(D[v],B[v])$ can be appended to $S[v]$ as a $(d,f)$ pair. It follows that once we iterate through all maximal flows we have retrieved all relevant $(d,f)$ pairs.\qed
\end{proof}

\begin{lemma}
Algorithm \ref{alg:unit} runs in $O(mn)$ worst case time.
\end{lemma}
\begin{proof}
We perform lifetime analysis to determine the upper bound of Algorithm \ref{alg:unit}. Each vertex $v$ can be cut from the SPT and be re-added to the SPT $O(n)$ times, once per each possible distance from $s$. Cutting/adding $v$ from/to the SPT takes $O(1)$ time, achieved by setting the parent of $v$ to either $NULL$ or $u$, respectively. Therefore the total time complexity of all operations involving the SPT is $O(n^2)$. Also there are a total of $O(n^{2})$ $(d,f)$ pairs. Before each vertex $v$ is added to the SPT all incoming edges $(u,v)$ are inspected. This results in $O(m)$ edges being inspected in total for the entire duration of the algorithm for each possible distance from $s$. Since there are $O(n)$ possible distances from $s$, the total time taken for edge inspection is $O(mn)$. Thus the total worst case time complexity becomes $O(n^{2} + mn) = O(mn)$.\qed
\end{proof}

\begin{theorem}
\label{theorem:int}
There exists an algorithm to solve SSSP-AF on directed graphs with unit edge costs in $O(mn)$ time.
\end{theorem}
\begin{proof}
There are $O(n)$ destination vertices, $O(m)$ maximal flows, and the length of each path is $O(n)$. Therefore storing all explicit paths as a solution to the SSSP-AF problem takes $O(mn^{2})$, which is too expensive. As is common in graph paths algorithms, we work around this problem by storing just the \emph{predecessor} vertex for storing the path information. The predecessor vertex for a path from $s$ to $v$ is the vertex that comes immediately before $v$ on the path.

We extend Algorithm \ref{alg:unit} to store the parent vertex, $u$, alongside the $(d,f)$ pair in line \ref{line:append} i.e. we can extend the $(d,f)$ pair to the $(d,f,u)$ triplet. Then $u$ is the predecessor vertex for the shortest path from $s$ to $v$ that can accommodate flow up to $f$. By using $d$ and $u$, any explicit path can be retrieved by recursively following the predecessor vertices in time linear to the path length. Clearly the additional storage of the predecessor vertex does not increase the worst case time complexity of the algorithm.\qed
\end{proof}

\subsection{Cascading Bucket System}
Before we move onto solving the SSSP-AF problem on graphs with integer edge costs, we review the $k$-level cascading bucket system (CBS) \cite{AMO,DF}. A detailed review of this data structure has also been provided by \citet{Tad12}.

In the $k$-level CBS the key value $d$ is given by:
$$
d = x_{k-1}p^{k-1} + x_{k-2}p^{k-2} + ... + x_1p^1 + x_{0}
$$
where $p$ is the number of buckets (or length) of each level.

Let $i$ be the largest index such that $x_{i}$ is non-zero. Then an element with key of $d$ is inserted into the $x_{i}^{th}$ bucket at level $i$. The values of $x_{i}$ for all $0 \leq i < k$ are calculated only once when an element is inserted, and each insertion takes $O(k)$ time.

The \emph{decrease-key} operation can be performed by removing the element from the CBS in $O(1)$ time, updating the key value, then re-inserting in the same level in $O(1)$ time, or re-inserting at a lower level in $O(l)$ time where $l$ is the difference between the initial level and the new level.

The \emph{delete-min} operation is more involved than the \emph{insert} or the \emph{decrease-key} operations. We maintain an active pointer at each level, $a_{i}$ for all $0 \leq i < k$, such that $a_{i}$ is the minimum index of the non-empty bucket at level $i$. $a_{i} = p$ means level $i$ is empty. To perform the \emph{delete-min} operation, if level 0 is not empty, we simply pick up the minimum non-empty bucket pointed to by $a_{0}$. If level 0 is empty, then we find the lowest non-empty level, $j$, and re-distribute the elements in the $a_{j}^{th}$ bucket into level $j-1$, then re-distribute the elements in the $a_{j-1}^{th}$ bucket into level $j-2$, and so on, until level 0 is non-empty. This process of repeated re-distribution from a higher level down to lower levels is referred to as \emph{cascading}, hence the name for the data structure. Each \emph{delete-min} takes $O(k+p)$ time if $j < k-1$, and $O(k+M/p^{k-1})$ if $j = k - 1$, where $M$ is the maximum key value that the CBS supports.

\begin{figure}
\def\svgwidth{235pt}
\begin{center}
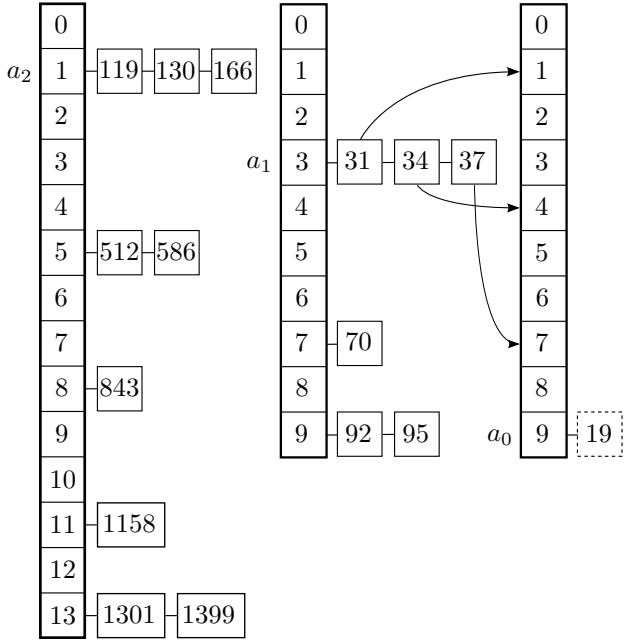
\end{center}
\caption{An example of a cascading bucket system with $k=3$ and $p=10$.}
\label{fig:cbs}
\end{figure}

\begin{example}
Figure \ref{fig:cbs} shows an example of a $3$-level CBS that can support key values up to 1399. $p = 10$ was chosen to make the example easier to understand, since it becomes straightforward to determine the correct bucket for base 10 numbers. If the element with key equal to 19 is removed from the CBS, $a_0$ becomes 10, and the next \emph{delete-min} operation will trigger the \emph{cascading} operation, resulting in elements in $a_{1}$ to be re-distributed into level 0.
\end{example}

\subsection{Integer edge costs}
\label{sec:ssspaf:int}
In Section \ref{sec:ssspaf:unit} we gave an algorithm to solve SSSP-AF on directed graphs with unit edge costs in $O(mn)$ time, that is faster than the straightforward method of $O(m^2)$. In this section we present two algorithms to solve SSSP-AF on directed graphs with non-negative edge costs in $O(m^2 + nc)$ and $O(m^2 + mn\log{(\frac{c}{m})})$ time bounds. Both time bounds are faster than the $O(m^2 + mn\log{\log{c}})$ time bound of the straightforward method for a wide range of values for $c$, $m$ and $n$, and have the added benefit of not relying on a complex data structure that is difficult to implement in real life situations. We note that Algorithm \ref{alg:unit} can also be used to solve SSSP-AF in $O(mnc)$ time since the maximum distance from $s$ for any vertex for any maximal flow value is $O(nc)$. $O(mnc)$ is a comparatively efficient time bound for dense graphs with small $c$.

The two time bounds of $O(m^2 + nc)$ and $O(m^2 + mn\log{(\frac{c}{m})})$ actually both come from the same algorithm, Algorithm \ref{alg:int}, but using different data structures to implement the priority queue. Algorithm \ref{alg:int} is a natural extension to the well known algorithm by \citet{Dijkstra}. For this algorithm we define the triplet $(v,d,f)$, where $v$ is the destination vertex and $d$ and $f$ are equivalent to the $(d,f)$ pair as defined in Section \ref{sec:ssspaf}. We let $D^{f}[v]$ be the current shortest distance from $s$ to $v$ for the maximal flow value of $f$. We let $Q$ be the priority queue for the $(v,d,f)$ triplets with $d$ as the key where the operations performed on $Q$ are \emph{insert}, \emph{decrease-key} and \emph{delete-min}. In summary, all $O(mn)$ $(v,d,f)$ triplets are added to the priority queue, $Q$, then as they are removed one by one, the $(d,f)$ pair is appended to $S[v]$ if the $(d,f)$ pair is \emph{useful} i.e. $f$ is greater than any flow in existing pairs in $S[v]$.

\begin{algorithm}
\caption{Solve SSSD-AF on graphs with integer edge costs}
\label{alg:int}
\begin{algorithmic}[1]
\algnotext{EndFor}
\algnotext{EndIf}
\algnotext{EndWhile}

\ForAll{$v \neq s \in V$}
	\ForAll{maximal flows $f$}
		\State{$D^{f}[v] \leftarrow \infty$}
	\EndFor
\EndFor

\State{$Q \leftarrow$ empty}
\State{Add $(s,0,\infty)$ to $Q$}
\ForAll{$v \neq s \in V$}
	\ForAll{maximal flow $f$}
		\State{Add $(v,\infty,f)$ to $Q$}
	\EndFor
\EndFor

\While{$Q$ is not empty}\label{line:while}
	\State{Delete $(v,d,f)$ with minimum $d$ from $Q$}\label{line:delmin}
	\ForAll{$w \neq s \in OUT(v)$}\label{line:out}
		\State{$f' \leftarrow \Call{min}{f, cap(v,w)}$}
		\State{$d' \leftarrow d + cost(v,w)$}
		\If{$d' < D^{f'}[w]$}
			\State{Let $x = (w,D^{f'}[w],f')$ in $Q$}
			\State{$D^{f'}[w] \leftarrow d'$}
			\State{Put $x$ in the correct position in $Q$}\label{line:update}
		\EndIf
	\EndFor
	\If{$S[v]$ is empty}
		\State{Append $(d,f)$ to $S[v]$}
	\Else
		\State{Let $(d_{0},f_{0}) \leftarrow$ last pair in $S[v]$}
		\If{$f_{0} < f$}
			\If{$d_{0} = d$}
				\State{Delete $(d_{0},f_{0})$ from $S[v]$}
			\EndIf
			\State{Append $(d,f)$ to $S[v]$}
		\EndIf
	\EndIf
\EndWhile

\end{algorithmic}
\end{algorithm}

\begin{definition}
$(d,f)$ in $S[v]$ is correct if $d$ is the shortest distance of a path that can push flows up to $f$ from $s$ to $v$.
\end{definition}

\begin{lemma}
\label{lem:sssdaf}
Algorithm \ref{alg:int} correctly solves SSSD-AF on directed graphs with integer edge costs.
\end{lemma}
\begin{proof}
We provide a formal proof by induction. Let $S$ be the set of $(v,d,f)$ such that $S[v]$ contains the pair $(d,f)$, for all $v$, for all $(d,f)$ pairs. Then in the beginning of each iteration:
\begin{enumerate}
\item The set of $(d,f)$ pairs in each $S[v]$ are all correct i.e. all $(v,d,f)$ triplets in $S$ are correct.
\item For any $(v,d,f)$ in $Q$, $d$ is the distance of the shortest path from $s$ to $v$ that can push $f$ only through the path that lies in $S$ except for the end point $v$.
\end{enumerate}
Basis. Before the while-loop begins (line \ref{line:while}) all are correct, and we suppose the theorem is correct at the beginning of some iteration. Then:
\begin{enumerate}
\item Let $(d_{0}, f_{0})$ be the last pair in $S[v]$. Suppose $(d,f)$ is appened to $S[v]$ at the end of the loop. Note that $d$ values are sorted in increasing order in $S[v]$. Since $f$ is appened only when $f > f_{0}$, there can be no shorter path in $S[v]$ that can push $f$. Thus $(d,f)$ is appended as a correct pair.
\item We remove $(v,d,f)$ from $Q$ in line \ref{line:delmin}. The new distance $d'$ and flow $f'$ from $v$ to $w$ are computed for all possible $w$. If $d' < D^{f}[w]$, $(v,d,f)$ now lies in the path from $s$ to $w$, and $(w,d',f')$ is added to $Q$. Since $(v,d,f)$ is added to $S[v]$ at the end of the loop, $d'$ is the distance of the shortest path from $s$ to $w$ that can push $f$ only through the path that lies in $S$ except for the end point $w$.\qed
\end{enumerate}
\end{proof}

\begin{lemma}
\label{lemma:nc}
Algorithm \ref{alg:int} can run in $O(m^{2} + nc)$ worst case time.
\end{lemma}
\begin{proof}
We use a one dimensional bucket system to implement Q. Insert and decrease-key operations can be performed in $O(1)$, resulting in $O(mn)$ time bound for both operations for the whole algorithm. The delete-min operation is performed simply by scanning through $Q$ from $i = 0$ to $nc$ one by one, where $i$ is the distance from $s$. We only have to scan through the distances once, and therefore the time complexity of the delete-min operation for the whole duration of the algorithm is $O(nc)$. Each vertex can be inspected exactly once at each maximal flow value. This means $O(m)$ edge inspections are performed per maximal flow value, resulting in $O(m^2)$ for the whole algorithm. Thus we have $O(m^{2} + mn + nc) = O(m^{2} + nc)$ as the total worst case time complexity of Algorithm \ref{alg:int} using the one dimensional bucket system to implement the priority queue.\qed
\end{proof}

\begin{lemma}
\label{lem:cm}
Algorithm \ref{alg:int} can run in $O(m^{2} + mn\log{(\frac{c}{m})})$ worst case time.
\end{lemma}
\begin{proof}
We use $k$-level CBS to implement $Q$. Since there are a total of $O(mn)$ $(v,d,f)$ triplets to be inserted into $Q$, the total time complexities for operations performed on $Q$ are: $O(kmn)$ for \emph{insert}, $O(kmn)$ for \emph{decrease-key}, and $O(kmn + pmn + cn/p^{k-1})$ for \emph{delete-min}, where p is the length of each level. We choose $p = (\frac{c}{m})^{1/k}$ and $k = \log{(\frac{c}{m})}$ to implement our $k$-level CBS. Then the term $O(kmn)$ dominates and the total time complexity for all three operations involving $Q$ becomes $O(mn\log{(\frac{c}{m})})$. As shown in the proof of Lemma \ref{lemma:nc}, the total time taken for edge inspection is $O(m^{2})$, resulting in the total worst case time complexity of $O(m^{2} + mn\log{(\frac{c}{m})})$ for Algorithm \ref{alg:int} using the $k$-level CBS to implement the priority queue.\qed
\end{proof}

\begin{theorem}
There exists algorithms to solve the SSSP-AF problem on directed graphs with integer edge costs in $O(m^{2} + nc)$ or $O(m^{2} + mn\log{(\frac{c}{m})})$ time bounds.
\end{theorem}
\begin{proof}
We take the same approach as discussed in the proof of Theorem \ref{theorem:int}. In line \ref{line:update} of Algorithm \ref{alg:int}, we store $v$ as the predecessor vertex alongside the $(w,d',f')$ triplet.\qed
\end{proof}

Note that in the proof of Lemma \ref{lem:cm} we could have chosen $k = O(\log{(\frac{c}{m})}/\log{\log{(\frac{c}{m})}})$ to speed up the algorithm by a polylog factor. We also note that if $c = 1$ then Algorithm \ref{alg:int} has the time complexity of $O(m^{2})$, thus Algorithm \ref{alg:unit} has not been made redundant by Algorithm \ref{alg:int}.

%%%%%%%%
%%%%%%%% NEW SECTION 
%%%%%%%%

\section{All Pairs Shortest Paths for All Flows}
The key achievement in this paper is not to come up with an original data structure but to devise algorithms to successfully utilize existing well known data structures, based on the observation that the maximum distance of any simple path on a graph with integer edge costs bounded by $c$ is $O(nc)$. From this observation what we have effectively achieved is to find a method to \emph{share resources}. That is, instead of having to repeatedly scan over $O(nc)$ distances for solving SSSP for each maximal flow value, we solve SSSP for all maximal flows at the same time while sharing the common resource, $Q$, thereby allowing us to scan $O(nc)$ only once. \citet{Tad12} used a similar idea to achieve $O(mn + n^{2}\log{(\frac{c}{n})})$ for the APSP problem. In this section we further extend the idea of sharing common resources to solve the problem of APSP-AF.

We let $D^{f}[u][v]$ be the currently known shortest distance from vertex $u$ to vertex $v$ for maximal flow $f$. We extend the $(v,d,f)$ triplet that was defined in Section \ref{sec:ssspaf:int} to the quadruple $(u,v,d,f)$, where $u$ and $v$ are the starting and the ending vertices of a possible path, respectively. Then we can extend Algorithm \ref{alg:int} to solve the APSD-AF problem, as shown in Algorithm \ref{alg:ap}

\begin{algorithm}
\caption{Solve APSD-AF on graphs with integer edge costs}
\label{alg:ap}
\begin{algorithmic}[1]
\algnotext{EndFor}
\algnotext{EndIf}
\algnotext{EndWhile}

\ForAll{$(u,v) \in V \times V$}
	\ForAll{maximal flows $f$}
		\If{u = v}
			\State{$D^{f}[u][v] \leftarrow 0$}
		\Else
			\State{$D^{f}[u][v] \leftarrow \infty$}
		\EndIf
	\EndFor
\EndFor

\State{$Q \leftarrow$ empty}
\ForAll{$v \in V$}
	\State{Add $(v,v,0,\infty)$ to $Q$}
\EndFor

\ForAll{$(u,v) \in V \times V$ such that $u \neq v$}
	\ForAll{maximal flow $f$}
		\State{Add $(u,v,\infty,f)$ to $Q$}
	\EndFor
\EndFor

\While{$Q$ is not empty}
	\State{Delete $(u,v,d,f)$ with minimum $d$ from $Q$}
	\ForAll{For all $w \neq u \in OUT(v)$}
		\State{$f' \leftarrow \Call{min}{f, cap(v,w)}$}
		\State{$d' \leftarrow d + cost(v,w)$}
		\If{$d' < D^{f'}[u][w]$}
			\State{Let $x = (u,w,D^{f'}[u][w],f')$ in $Q$}
			\State{$D^{f'}[u][w] \leftarrow d'$}
			\State{Put $x$ in the correct position in $Q$}\label{line:ap:update}
		\EndIf
	\EndFor
	\If{$S[u][v]$ is empty}
		\State{Append $(d,f)$ to $S[u][v]$}
	\Else
		\State{Let $(d_{0},f_{0}) \leftarrow$ last pair in $S[u][v]$}
		\If{$f_{0} < f$}
			\If{$d_{0} = d$}
				\State{Delete $(d_{0},f_{0})$ from $S[u][v]$}
			\EndIf
			\State{Append $(d,f)$ to $S[u][v]$}
		\EndIf
	\EndIf
\EndWhile

\end{algorithmic}
\end{algorithm}

\begin{lemma}
Algorithm \ref{alg:ap} correctly solves APSD-AF on directed graphs with integer edge costs.
\end{lemma}
\begin{proof}
Essentially the same argument as the proof of Lemma \ref{lem:sssdaf} can be applied. The only differences are that we now have quadruples $(u,v,d,f)$ instead of triplets in $Q$ and $S$, and we need to go through more iterations as we are solving for all $O(n^{2})$ pairs of vertices. Clearly these differences has no impact on the correctness of the algorithm.\qed
\end{proof}

\begin{lemma}
Algorithm \ref{alg:ap} can run in $O(m^{2}n + nc)$ worst case time.
\end{lemma}
\begin{proof}
There are $O(n^{2})$ pairs of vertices resulting in $O(mn^{2})$ $(u,v,d,f)$ quadruples. Each vertex pair can be observed exactly once at each maximal flow value hence the number of edge inspections that occur at one maximal flow value is $O(mn)$, resulting in the total time bound of $O(m^{2}n)$ for edge inspections. Using the one dimensional bucket system to implement $Q$, we have $O(m^{2}n + mn^{2} + nc) = O(m^{2}n + nc)$ as the worst case time complexity.\qed
\end{proof}

\begin{lemma}
Algorithm \ref{alg:ap} can run in $O(m^{2}n + mn^{2}\log{(\frac{c}{mn})})$ worst case time.
\end{lemma}
\begin{proof}
We use the $k$-level cascading bucket system to implement $Q$, where the \emph{delete-min} operation now takes $O(kmn^{2} + pmn^{2} + cn/p^{k-1})$ time. We choose $p = (\frac{c}{mn})^{1/k}$ and $k = \log{(\frac{c}{mn})}$ for the total time complexity of the \emph{delete-min} operation to become $O(mn^{2}\log{(\frac{c}{mn})})$. Thus the total worst case time complexity using the cascading bucket system is $O(m^{2}n + mn\log{(\frac{c}{mn})})$.\qed
\end{proof}

\begin{theorem}
There exists algorithms to solve the APSP-AF problem on directed graphs with integer edge costs in $O(m^{2}n + nc)$ or $O(m^{2}n + mn^{2}\log{(\frac{c}{mn})})$ time bounds.
\end{theorem}
\begin{proof}
We take the same approach as before and modify Algorithm \ref{alg:ap} to store $v$ as the predecessor vertex alongside the $(u,w,d',f')$ quadruple in line \ref{line:ap:update}.\qed
\end{proof}

%%%%%%%%
%%%%%%%% NEW SECTION 
%%%%%%%%

\section{Concluding remarks}
We have introduced a new graph path problem and provided non-trivial algorithms to solve the new problem that are both practical and faster than the straightforward methods.

The example of evacuation planning has been used in the introduction of the paper to show the relevance of the SP-AF problem in real life. Another possible application of the SSSP-AF problem is in computer networking, as a more sophisticated dynamic routing protocol than the currently commonly used protocols such as RIP and OSPF. And with the introduction of Software Defined Networking (SDN) \cite{ONF}, we can also propose APSP-AF as a possible algorithm to calculate the routes in the entire network.

Trivial lower bounds of $O(mn)$ and $O(mn^{2})$ exist for the SSSP-AF and APSP-AF problems, respectively, on weighted digraphs. This paper has investigated only graphs with integer edge costs. Can we provide a better time bound than the straightforward $O(m^{2} + mn\log{n})$ for the SSSP-AF problem on directed graphs with real edge costs? Is there a faster algorithm on undirected graphs? How close can we get to the lower bounds of the SP-AF problems? We conclude the paper with these open questions and look forward to further research that may address these open problems.

\bibliographystyle{agsm}    % or some other suitable package.

%\bibliography{CRPITExample}  % often included from a separate file.

\end{document}